\documentclass[11pt,oneside]{article}
\usepackage{import}
\usepackage{etoolbox}

\usepackage[font=small,skip=0pt]{caption}

\providebool{DRAFT}
\booltrue{DRAFT}

\usepackage{etoolbox}
\usepackage{xparse}

\makeatletter%
\@ifclassloaded{book}
{%
  \newcommand{\whenbook}[1]{%
    #1%
  }%
}
\makeatother%

\providecommand{\whenbook}[1]{}%

\providebool{bookdraft}%

\NewDocumentEnvironment{When}{m +b}{%
  \providebool{#1}%
  \ifbool{#1}{#2}{}%
}{}%

\NewDocumentEnvironment{Unless}{m +b}{%
  \providebool{#1}%
  \ifbool{#1}{}{#2}%
}{}%

\providebool{DRAFT}
\ifbool{DRAFT}{
  \newcommand{\whendraft}[1]{#1}%
    \newcommand{\unlessdraft}[1]{}%
  }

\NewDocumentEnvironment{Draft}{+b}{%
  \whendraft{#1}%
}{%
}%

\NewDocumentEnvironment{DRAFT}{O{BlueViolet}}{%
  \begin{Draft}%
    \color{#1}%
  }{%
  \end{Draft}
}%

\usepackage[hyphens]{url}

\usepackage[style=alphabetic,natbib=true,maxalphanames=4,minalphanames=3,maxnames=10,doi=true,isbn=false,url=false,backend=biber]{biblatex} %

\AtEveryBibitem{%
\ifentrytype{proceedings}{
    \clearfield{editor}%
    \clearname{editor}%
  }{}
  \ifentrytype{inproceedings}{
      \clearfield{editor}%
      \clearname{editor}%
}{}
}

\usepackage{amsmath}                               %
\usepackage{amsthm}                                %
\usepackage[anchorcolor=blue,colorlinks]{hyperref}%
\usepackage{varioref}
\usepackage[capitalize]{cleveref}

\crefname{invariantsi}{invariant}{invariants}

\usepackage[margin=1in]{geometry}

\usepackage{fancyhdr}

\usepackage[spacing,kerning=true]{microtype}          %
\usepackage[T1]{fontenc}

\usepackage{article} %
\usepackage{math} %
\usepackage{wrapfig} %
\usepackage{placeins} %
\usepackage{advdate}%

\usepackage{appendix}

\usepackage{graphicx} %
\usepackage{layout} %
\usepackage{setspace} %
\usepackage{mathtools} %
\usepackage{grffile} %
\usepackage{pdfpages} %
\usepackage{mdframed} %
\usepackage{bm} %
\usepackage{needspace} %
\usepackage{pbox} %
\usepackage{cancel}
\usepackage[normalem]{ulem}

\usepackage{tikz}        
\usetikzlibrary{arrows.meta,calc}

\usepackage[full]{textcomp}     %
\usepackage{lmodern}            %
\usepackage[T1]{fontenc}        %
\usepackage{inconsolata}        %

\usepackage{titlesec}

\usepackage{truncate}

\NewDocumentCommand{\undernote}{s O{blue} m m}{
  \IfBooleanT{#1}{\smash}%
  {\color{#2} %
    \underbrace{\normalcolor%
      #4}_{\mathclap{\text{#3}}} %
  }%
  \IfBooleanT{#1}{\vphantom{#4}}
}

\newcommand{\defterm}[1]{{\boldmath\normalfont \bfseries #1}}%
\renewcommand{\defterm}{\emph}%

\makeatletter
\g@addto@macro\bfseries{\boldmath}
\makeatother

\titlespacing*{\paragraph}{%
  0pt}{
  {\medskipamount}}{
  1em}

\titleformat{\subparagraph}[runin]{\itshape}{0pt}{}{}%

\titlespacing*{\subparagraph}{%
  0pt}{
  {\medskipamount}}{
  1em}

\newlength{\myalphabet}                %
\newlength{\mywidth}                   %
\newlength{\mymargin}                  %

\providecommand{\wrt}{with respect to\xspace}%
\usepackage{skak}%

\providecommand{\poly}{\fparnew{\operatorname{poly}}}

\definecolor{calypso}{RGB}{50, 104, 145} %

\hypersetup{allcolors=calypso}  %

\definecolor{almostblack}{RGB}{18, 18, 18} %

\color{almostblack}%
\hypersetup{allcolors=almostblack}  %

\providebool{article}
\booltrue{article}

\begin{document}

\newcommand{\W}{\mathcal{W}}
\newcommand{\potential}{\fparnew{\varphi}} %
\providecommand{\pot}{\potential}       %
\providecommand{\pote}{\potential}
\providecommand{\potentialB}{\fparnew{\psi}}      %
\providecommand{\poteB}{\potentialB}              %
\providecommand{\len}{\fparnew{\ell}}       %
\newcommand{\potlen}{\fparnew{\ell_{\potential}}} %
\newcommand{\plen}{\fparnew{\ell_{\potential}}}   %
\newcommand{\p}{\fparnew{\varphi}}                %
\newcommand{\lenp}{\fparnew{\ell_{\p}}}

\newcommand{\through}{\fparnew{\operatorname{T}}} %
\newcommand{\sandwich}{\fparnew{\operatorname{S}}} %
\newcommand{\disp}{\fparnew{d_{\p}}}                 %
\newcommand{\dis}{\fparnew{d}}                     %
\newcommand{\hopd}[1]{\fparnew{d^{#1}}}           %
\newcommand{\hopdp}[2][\p]{\fparnew{d^{#2}_{#1}}}            %
\newcommand{\shopd}[1]{\fparnew{\hat{d}^{#1}}}               %
\newcommand{\shopdp}[2][\p]{\fparnew{\hat{d}_{#1}^{#2}}}
\newcommand{\phopd}[1]{\fparnew{\hat{d}^{#1}}}               %
\newcommand{\phopdp}[2][\p]{\fparnew{\hat{d}_{#1}^{#2}}}
\providecommand{\deg}{\fparnew{\operatorname{deg}}} %
\providecommand{\indeg}{\fparnew{\operatorname{deg}^-}} %
\providecommand{\outdeg}{\fparnew{\operatorname{deg}^+}} %
\providecommand{\hopindeg}[1]{\fparnew{\operatorname{deg}_{-}^{#1}}} %
\providecommand{\hopoutdeg}[1]{\fparnew{\operatorname{deg}_{-}^{#1}}} %
\providecommand{\heavyin}[1]{H^{-}_{#1}} %
\providecommand{\hin}{\heavyin} %
\providecommand{\heavyout}[1]{H^{+}_{#1}} %
\providecommand{\hout}{\heavyout}
\providecommand{\reachin}[1]{R^{-}_{#1}} %
\providecommand{\reachout}[1]{R^+_{#1}}  %
\providecommand{\apxheavy}{\smash{\tilde{H}}} %
\providecommand{\negV}{N}                     %
\providecommand{\negE}{E^-}
\providecommand{\numN}{k}
\renewcommand{\mod}{\operatorname{mod}} %
\providecommand{\distance}{\fparnew{d}} %

\providecommand{\varh}{\eta}             %
\providecommand{\varvarh}{\zeta}
\providecommand{\supersource}{s^{\star}} %
\NewDocumentCommand{\envelope}{m m m e{_}}{
  B^{#1}%
  \IfNoValueF{#4}{_{#4}}%
  \parof{#2,#3}
}
\NewDocumentCommand{\negativeenvelope}{m m m e{_}}{
  \bar{B}^{#1}%
  \IfNoValueF{#4}{_{#4}}%
  \parof{#2,#3}
}
\newcommand{\env}{\envelope}%
\newcommand{\nenv}{\negativeenvelope}%
\NewDocumentCommand{\subhopd}{m}{\fparnew{D^{#1}}}
\NewDocumentCommand{\subpopd}{m}{\fparnew{\hat{D}^{#1}}}
\renewcommand{\bar}{\overline}%
\providecommand{\Heads}{\bar{U}}
\providecommand{\outcut}{\fparnew{\partial^+}}

\renewcommand{\subsection}[1]{\paragraph{#1.}}


\newcommand{\authornote}{%
  \texttt{$\setof{\texttt{huan1754},\texttt{jin453},\texttt{krq}}$@purdue.edu}.
  Purdue University, West Lafayette, Indiana. YH was supported in part
  by NSF grant IIS-2007481. PJ and KQ were supported in part by NSF
  grant CCF-2129816.}

\author{Yufan Huang \and Peter Jin \and Kent Quanrud}

\title{Faster negative length shortest paths
  by~bootstrapping~hop~reducers}

\maketitle

\begin{abstract}
  The textbook algorithm for real-weighted single-source shortest
  paths takes $\bigO{m n}$ time on a graph with $m$ edges and $n$
  vertices. The breakthrough algorithm by \citet{Fineman24} takes
  $\apxO{m n^{8/9}}$ randomized time. The running time was
  subsequently improved to $\apxO{mn^{4/5}}$ \cite{HJQ25}.

  We build on \cite{Fineman24,HJQ25} to obtain an
  $\apxO{m n^{3/4} + m^{4/5} n}$ randomized running
  time. (Equivalently, $\apxO{mn^{3/4}}$ for $m \geq n^{5/4}$, and
  $\apxO{m^{4/5} n}$ for $m \leq n^{5/4}$.) The main new technique
  replaces the hop-reducing auxiliary graph from \cite{Fineman24} with
  a bootstrapping process where constant-hop reducers for small
  subgraphs of the input graph are iteratively amplified and expanded
  until the desired polynomial-hop reduction is achieved over the
  entire graph.
\end{abstract}

\unmarkedfootnote{\authornote}

\section{Introduction}

Let $G = (V,E)$ be a directed graph with $m$ edges, $n$ vertices, and
real-valued edge lengths $\len: E \to \reals$. For a fixed source
vertex $s$, the single-source shortest path (SSSP) problem is to
compute either a negative length cycle in $G$, or the shortest-path
distance from $s$ to every other vertex in $G$.  While there are
nearly linear time algorithms for nonnegative edge weights
\cite{Dijkstra59,FT87,DMMSY25} and integer weights (\wrt the bit
complexity) \cite{BNW22,BCF23}, algorithms for the real-weighted
setting are far from linear.

The textbook $\bigO{mn}$ time dynamic programming algorithm dates back
to the 1950s \cite{Shimbel55,Ford56,Bellman58,Moore59} and remained
the fastest algorithm (for any edge density) until very recently. The
first improvement was a celebrated $\apxO{mn^{8/9}}$ randomized time
algorithm by \citet{Fineman24}. ($\apxO{\cdots}$ hides logarithmic
factors.) The running time was subsequently improved to
$\apxO{m n^{4/5}}$ in \cite{HJQ25}.
\begin{theorem}
  The SSSP problem for real-weighted graphs can be solved with high
  probability in $\apxO{m n^{3/4} + m^{4/5} n}$ randomized time.
\end{theorem}

An equivalent way to state the running time is $\apxO{mn^{3/4}}$ for
$m \geq n^{5/4} $ and $\apxO{m^{4/5} n}$ for $m \leq
n^{5/4}$.

We begin with preliminaries in
\Refsection{preliminaries}. \Refsection{previous-algorithms} outlines
more involved definitions and subroutines from previous works
\cite{Fineman24,HJQ25}, and contains an overview of how the components
fit together to obtain the $\apxO{mn^{4/5}}$ running time from
\cite{HJQ25}. \Refsection{bootstrap} presents a bootstrapping
construction of ``hop reducers'' and is the main technical
contribution of this work. \Refsection{dense} applies the new
techniques to denser graphs ($m \geq \apxOmega{n^{5/4}}$) and
\refsection{sparse} applies them to sparser graphs
($m \leq \apxO{n^{5/4}}$).

\section{Preliminaries}

\labelsection{preliminaries}

Let $\mu = m + n \log n$ and let $\numN$ denote the number of negative
edges.

\subsection{Distances}
The distance from $s$ to $t$, denoted by $d(s,t)$, is defined as the
infimum length over all walks from $s$ to $t$.  For vertex
sets $S$ and $T$, we let $d(S,T)$ denote the infimum distance $d(s,t)$
over all $s \in S$ and $t \in T$.

\subsection{Preprocessing and negative vertices}

We assume the maximum in-degree and out-degree are both
$\bigO{m / n}$. We also assume that $G$ has $\numN \leq n/2$ negative
edges, and that each negative edge $(u,v)$ is the unique incoming edge
of its head $v$ and the unique outgoing edge of its tail $u$
\cite{Fineman24,HJQ25}.  We identify each negative edge $(u,v)$ with
its tail $u$, and call $u$ a \defterm{negative vertex}.  When there is
no risk of confusion, we may reference a negative edge by its negative
vertex and vice-versa. We let $N$ denote the set of negative
vertices. For a set of negative vertices $U \subseteq N$, we let
$\Heads = \setof{v \where u \in U \andcomma (u,v) \in E}$ be the set
of heads of the negative edges associated with $U$.

For a set of negative vertices $U$, we let $G_U$ denote the subgraph
obtained by restricting the set of negative edges to those
corresponding to $U$. We let $d_U(\cdot,\cdot) = d_{G_U}(\cdot,\cdot)$
denote distances in $G_U$.  We let $G^+ = G_{\emptyset}$ denote the
subgraph of nonnegative edges.

\subsection{Hop distances}

The number of \emph{hops} in a walk is the number of negative edges
along the walk, counted with repetition. An \defterm{$h$-hop walk} is
a walk with at most $h$ negative edges. For an integer $h$, and
vertices $s,t \in V$, we let
\begin{math}
  \hopd{h}{s,t}
\end{math}
denote the infimum length over all $h$-hop walks from $s$ to $t$.
Hop distances obey the recurrence
\begin{align*}
  \hopd{h+1}{s,t} = \min{\hopd{h}{s,t},\, \min_{u \in N} \hopd{h}{s,u} +
  \len{u,v}+
  \hopd{0}{v,t}}. \labelthisequation{hop-distance}
\end{align*}
For fixed $s$, if $\hopd{h+1}{s,v} = \hopd{h}{s,v}$ for all
$v \in \negV$, then $\dis{s,v} = \hopd{h+1}{s,v}$ for all $v$.

For fixed $s$, and $h \in \naturalnumbers$, one can compute the
$h$-hop distance $\hopd{h}{s,v}$ for all $v \in V$ in $\bigO{h \mu}$
time by a hybrid of Dijkstra's algorithm and dynamic programming
\cite{DI17,BNW22}. The algorithm can be inferred from
\refequation{hop-distance}: given $\hopd{h}{s,u}$ for all $u$, we can
compute $\min_{u \in N} \hopd{h}{s,u} + \len{u,v} + \hopd{0}{v,t}$ for
all $t$ with a single call to Dijkstra's algorithm over an appropriate
auxiliary graph. By maintaining parent pointers, this also computes
the walks attaining the $h$-hop distances in the same running time.

\subsection{Proper hop distances}

A \emph{proper $h$-hop walk} is defined as a walk with exactly $h$
negative edges and where all negative vertices are distinct.  For
vertices $s,t$, let
\begin{math}
  \shopd{h}{s,t}
\end{math}
denote the infimum length over all proper $h$-hop walks from $s$ to
$t$.  We call
\begin{math}
  \shopd{h}{s,t}
\end{math}
the \defterm{proper $h$-hop distance} from $s$ to $t$.

Proper hop distances are generally intractable.  For large $h$, proper
hop distances capture NP-Hard problems such as Hamiltonian path.  For
small $h$, single-source proper $h$-hop distances can be computed in
$\apxO{2^{\bigO{h}} m}$ randomized time using color-coding techniques
\cite{AYZ95}. We will only indirectly access proper hop distances via
the following efficient subroutine.

\begin{lemma}[\cite{HJQ25}, Lemma 3.1]
  \labellemma{implicit-proper-distance} Given a set of negative
  vertices $S$, there is an $\bigO{h \mu}$-time algorithm that returns
  either a negative cycle, a pair $s,t \in S$ with
  $\shopd{h}{s,t}_S < 0$, or the distances $d_S(V,t)$ for all
  $t \in V$.
\end{lemma}

\subsection{Vertex potentials}

Given vertex potentials $\p: V \to \reals$, let
$\lenp{e} \defeq \len{e} + \p{u} - \p{v}$ for an edge $e = (u,v)$. We
say that $\p$ \emph{neutralizes} a negative edge $e$ if
$\lenp{e} \geq 0$, and neutralizes a negative vertex if it neutralizes
the corresponding negative edge. \citet{Johnson77} observed that the
potentials $\p{v} = d(V,v)$ neutralizes all edges (assuming there are
no negative-length cycles). We let $G_{\varphi}$ denote the graph with
edges reweighted by $\p$ and $\disp$ denote distances \wrt $\lenp$.

We say that potentials $\p$ are \emph{valid} if $\lenp{e} \geq 0$ for
all $e$ with $\len{e} \geq 0$; i.e., $\p$ does not introduce any new
negative edges. For any hop parameter $h \in \naturalnumbers$ and any
set of vertices $S$, the potentials $\p{v}_1 = \hopd{h}{S,v}$ and
$\p{v}_2 = -\hopd{h}{v, S}$ are valid.  Given two valid potentials
$\p_1$ and $\p_2$, $\max{\p_1,\p_2}$ and $\min{\p_1,\p_2}$ are also
valid potentials. If $\p_1$ is valid for $G$, and $\p_2$ is valid for
$G_{\p_1}$, then $\p_1 + \p_2$ is valid for $G$.

The algorithms discussed here follow an incremental reweighting
approach, dating back to \cite{Goldberg95}, where the graph is
repeatedly reweighted along valid potentials that neutralize a subset
of negative edges, until (almost) all edges are nonnegative. (A small
number of remaining negative edges can always be cleaned up by
Johnson's technique in nearly linear time per negative edge.)  Then
the distances can be computed by Dijkstra's algorithm in the
reweighted graph.  Henceforth we focus exclusively on the effort to
neutralize all or almost all of the negative edges.

When algorithms reweight the graph along multiple potentials in one
iteration, we treat any initially negative edge as negative throughout
the iteration, even if incidentally neutralized. This is important for
preserving hop-counting invariants. $k$ remains the number of negative
edges at the beginning of the iteration.

\section{Previous algorithms and techniques}

\labelsection{previous-algorithms}

\subsection{Remote edges and hop reduction}

Given a set of negative vertices $U$, computing Johnson's potentials
on $G_U$ requires a $\sizeof{U}$-hop distance computation and thus
nearly linear time per vertex in $U$. \cite{Fineman24} introduced
auxiliary ``hop reduction'' graphs to decrease the number of hops that
are needed. We formalize the idea in the following definition.

\begin{definition}
  Let $h \in \naturalnumbers$. We say that a graph $H = (V_H,E_H)$ is
  an \defterm{$h$-hop reducer} for a graph $G = (V, E)$ if $V \subseteq V_H$ and
  \begin{math}
    \distance{s,t}_{G} \leq \hopd{\roundup{k/h}}{s,t}_H \leq \hopd{k}{s,t}_{G}
  \end{math}
  for all $s,t \in V$ and $k \in \naturalnumbers$.
\end{definition}

A natural idea for a hop reducer for $G_U$ is the following layered
auxiliary graph $H$. We start with $h+1$ copies of the positive part
$G^+$. Let $V_i$ denote the vertices in the $i$th layer (for
$i = 0,\dots,h$) and for $v \in V$ let $v_i$ denote the copy of $v$ in
$V_i$. For each negative arc $(u,v)$, we have copies $(u_i,v_{i+1})$
from each layer to the next. For each vertex $u$, we also have length
$0$ ``self-arcs'' $(u_i,u_{i+1 \mod h + 1})$ from each layer to the
next ($\mod h+1$).  Define a potential function $\varphi$ by
$\varphi(v_i) = d^i(V,v)$. It is easy to verify that $\pote$
neutralizes all edges in the auxiliary graph except the self-arcs
wrapping around from layer $h$ back to layer $0$.  Starting from
$V_0$, one hop in $H_{\varphi}$ (cycling through all the layers)
simulates $h$ hops in $G$, so one needs only $(\sizeof{U}/h)$-hop
distances in $H_{\varphi}$ to compute Johnson's potentials
neutralizing all the edges in $G_U$. The only catch is that $H$ is a
factor $h$ larger than $G_U$, so there is no overall speedup.

To reduce the size of $H$, \cite{Fineman24} introduced the notion of
``remote'' sets of negative vertices.  We say that a vertex $s$ can
\defterm{$h$-hop negatively reach} another vertex $t$ if there is an
$h$-hop $(s,t)$-walk with negative length.  For a parameter
$r \in \naturalnumbers$, and a set of negative vertices
$U \subseteq V$, $U$ is \defterm{$r$-remote} if the collective $r$-hop
negative reach of $U$ has at most $n/r$ vertices. Remote sets yield
linear-size hop reducers, as follows.

\begin{lemma}[{\cite[Lemma 3.3]{Fineman24}}]
  \labellemma{hop-reduction} Let $G$ have maximum in-degree and
  out-degree $\bigO{m / n}$.  Given a set of $r$-remote vertices $U$,
  one can construct an $r$-hop reducer for $G_U$ with $\bigO{m}$ edges
  and $\bigO{n}$ vertices in $\bigO{\mu}$ time.
\end{lemma}
We sketch the construction for comparison with our new techniques. Let
$X$ be the $r$-hop negative reach of $U$.  We modify the layered graph
$H$ described above (for $h = r$) so that layers $1$ through $r$ only
duplicate $G_X^+$.  Letting $\outcut{X}$ denote the directed cut of
arcs from $X$ to $V \setminus X$, we also add ``exit arcs''
$(x_i,v_0)$ for each arc $(x,v) \in \outcut{X}$ and each layer
$i > 0$.

One can verify that the same potentials $\varphi(v_i)= \hopd{i}{V,v}$
again neutralize all the arcs except the self-arcs wrapping around
from the last layer to the first. An important observation here is
that the exit arcs $(x_i,v_0)$ (where $(x,v) \in \outcut{X}$) remain
nonnegative because $v$ is not in the negative $i$-hop reach of $U$
for $i \leq h$.  Each hop in $H_{\varphi}$ simulates $r$ hops in $G$,
as before. This time, however, $H_{\pote}$ has size $\bigO{m}$ because
we are only duplicating the subgraph induced by the negative reach $X$
of size $\bigO{m /r}$. Consequently it takes only
$\apxO{m \sizeof{U} / r}$ time in $H_{\varphi}$ to compute the desired
Johnson's potentials for $G_U$.

We note that $H_{\varphi}$ can be computed in $\bigO{\mu}$ time,
because the potential function $\pote{v_i} = \hopd{i}{V,v}$ only
pertains to the negative reach $X$ consisting of at most $n/r$ vertices.

\subsection{Betweenness and sandwiches}

The remaining challenge is constructing remote sets. \cite{Fineman24}
showed how to create remote sets using two concepts: betweenness and
negative sandwiches.  For a parameter $h \in \naturalnumbers$, and
$s,t \in V$, the \defterm{$h$-hop betweenness} is defined as the
number of vertices $v$ such that $\hopd{h}{s,v} + \hopd{h}{v,t} <
0$. \cite{Fineman24} gave the following randomized procedure to reduce
the betweenness of all pairs of vertices.

\begin{lemma}[{\cite[Lemma 3.5]{Fineman24}}]
  \labellemma{betweenness-reduction} For $h \in \naturalnumbers$,
  there is an $\bigO{\mu h^2 \log n }$ time randomized algorithm that
  returns either a set of valid potentials $\p$ or a negative cycle.
  With high probability, all pairs $(s,t) \in V \times V$ have
  $\bigO{h}$-hop betweenness (at most) $n/h$ \wrt $\ell_{\varphi}$.
\end{lemma}

\cite{Fineman24} defines a \defterm{negative sandwich} as a triple
$(s,U,t)$, where $s,t \in V$ and $U \subseteq N$, such that
$\hopd{1}{s,u} < 0$ and $\hopd{1}{u,t} < 0$ for all $u \in U$.
\cite{HJQ25} generalized this by defining a \defterm{weak $h$-hop
  negative sandwich}, for $h \in \naturalnumbers$, as a triple
$(s,U,t)$ such that
\begin{math}
  \hopd{h}{s,u} + \hopd{h}{u,t} \leq 0
\end{math}
for all $u \in U$.  The following lemma converts sandwiches into
remote edges when the betweenness is low. The lemma modifies
\cite[Lemma 2.3]{HJQ25} so that the potential $\p$ only depends on $h$
but the proof is the same.
\begin{lemma}
  \labellemma{sandwich->remote}\labellemma{remotization} Let $(s,U,t)$
  be a weak $h$-hop negative sandwich and let $(s,t)$ have $(\eta+h)$-hop
  betweenness $n/\ell$. Let
  $\p{v} = \min{\hopd{h}{s,v}, - \hopd{h}{v,t}}$. Then the $\eta$-hop
  negative reach of $U$ in $G_{\varphi}$ has at most $n/\ell$
  vertices.
\end{lemma}

The final ingredient is a method for generating sandwiches. The
following subroutine from \cite{HJQ25} combines
\reflemma{implicit-proper-distance} with randomized sampling.

\begin{lemma}[\cite{HJQ25}, Lemma 4.1]
  \labellemma{sampling} For $h = \bigOmega{\log n}$, there is a
  randomized algorithm that, in $\bigO{h \mu}$ randomized time,
  returns either a negative cycle, a weak $h$-hop sandwich $(s,U,t)$,
  or a set of negative vertices $S$ and the distances $d_S(V,v)$ for
  all vertices $v$. With high probability, we have
  $\sizeof{U}, \sizeof{S} \geq \bigOmega{\sqrt{h k}}$ (when they are
  returned by the algorithm).
\end{lemma}

\subsection{$\apxO{mn^{4/5}}$ randomized time}
We briefly review the $\apxO{mn^{4/5}}$ randomized time algorithm of
\cite{HJQ25}, assuming for simplicity that there are no negative
cycles.  We describe a single iteration of the algorithm that
neutralizes a subset of negative edges in $\apxO{m / k^{1/5}}$
amortized time per negative vertex. The overall running time follows
from repeating this iteration.

Let $h = k^{1/5}$. We first reduce the $h$-hop betweenness to $n/h$ in
$\apxO{mh^2}$ time (\reflemma{betweenness-reduction}). We then either
directly neutralize a set of $\apxOmega{\sqrt{hk}}$ negative edges or
obtain a weak $h$-hop sandwich $(s,U,t)$ of size
$\apxOmega{\sqrt{hk}}$ (\reflemma{sampling}), where
$\sqrt{hk} = k^{3/5}$. In the latter case, we reweight $G_U$ to make
$U$ $h$-remote (\reflemma{remotization}), and construct an $h$-hop
reducer $H$ for $G_U$ with $\bigO{m}$ edges and $\bigO{n}$ vertices
(\reflemma{hop-reduction}). We compute Johnson's potentials for $G_U$
via $H$ in $\bigO{\mu \sizeof{U}/h}$ time. Overall, we spend
$\apxO{m k^{2/5}}$ time to neutralize $\apxOmega{k^{3/5}}$ negative
vertices.

\section{Bootstrapping hop reducers}

\labelsection{bootstrap}

A bottleneck in \cite{HJQ25} (as well as in \cite{Fineman24}) is the
$\apxO{mh^2}$ time spent to reduce the $\bigO{h}$-hop betweenness to
$n/h$ in \reflemma{betweenness-reduction}. Recall that the reason that
we reduce the $\bigO{h}$-hop betweenness to $n/h$ is so that negative
sandwiches can be converted into $h$-remote sets, for which we have
the linear-size $h$-hop reducer of \reflemma{hop-reduction}.

This section describes a process that creates an $h$-hop reducer even
when the $h$-hop betweenness is only the trivial bound of $n$. This
procedure instead requires the $\parof{\eta + \bigO{\log n}}$-hop
betweenness to be $n \eta /h$ for $\eta \leq h$.  These conditions can
be accomplished in just $\apxO{m h + nh^2}$ time as follows.

\begin{lemma}
  \labellemma{multiscale-betweenness-reduction} For
  $h = \bigOmega{\log n}$, in $\bigO{(m + nh) h \log^2 n}$ time, one
  can compute potentials $\varphi$ with the following property: for
  all $s,t \in V$, and $\varh \in \naturalnumbers$, $(s,t)$ has
  $\parof{\varh + \bigO{\log n}}$-hop betweenness at most
  $n \eta / h$.
\end{lemma}
\begin{proof}
  The proof applies the same techniques as
  \reflemma{betweenness-reduction} to every scale from $1$ to $h$. We
  include the proof for the sake of completeness.  Let
  $L = \logup[2]{h}$. For $i \in \setof{0,\dots,L}$, let
  $\varh_i = 2^i + \bigO{\log n}$, and let $R_i$ sample
  $\bigO{2^{L - i} \log n}$ vertices uniformly at random.

  We construct an auxiliary graph $H$ starting from $G^+$. For
  $i \in \setof{0,\dots,L}$, $\varh_i = 2^i + \bigO{\log n}$,
  $v \in V$, and $r \in R_i$, we add arcs $(v,r)$ and $(r,v)$ of
  length $\hopd{\varh_i}{v,r}$ and $\hopd{\varh_i}{r,v}$,
  respectively.

  Let $\varphi$ be the Johnson's potentials for $H$. Since $H$
  contains $G^+$, $\varphi$ is valid for $G$. Moreover, for all
  $r \in R_i$ and $v \in V$, we have
  $\hopdp{\varh_i}{r,v} = \varphi(r) + \hopd{\varh_i}{r,v} -
  \varphi(v) = \varphi(r) + \len{r,v}_H - \varphi(v) \geq 0$, and
  similarly $\hopdp{\varh_i}{v,r} \geq 0$.

  To verify the $r$-hop betweenness, fix $s,t \in V$ and a hop
  parameter $\varh \leq h + \bigO{\log n}$. Let
  $i = \logup[2]{\varh}$. Rank all the vertices $x \in V$ in
  increasing order of $\hopd{\varh_i}{s,x} + \hopd{\varh_i}{x,t}$.
  With high probability, $R_i$ contains a vertex $x$ with rank at most
  $n/2^{L-i+1}$ between $s$ and $t$. Then
  $\hopdp{\varh_i}{s,x} + \hopdp{\varh_i}{x,t} \geq 0$. For all other
  vertices $y$ of higher rank, we have
  \begin{align*}
    \hopdp{\varh}{s,y} + \hopdp{\varh}{y,t}
    &\geq
      \hopdp{\varh_i}{s,y} + \hopdp{\varh_i}{y,t}
      = \pote{s} + \hopd{\varh_i}{s,y} + \hopd{\varh_i}{y,t}
      -\pote{t} \\
    &\geq \pote{s} + \hopd{\varh_i}(s,x) +
      \hopd{\varh_i}{x,t} - \pote{t}
      =  \hopdp{\varh_i}{s,x} + \hopdp{\varh_i}{x,y}
      \geq 0.
  \end{align*}
  Thus $(s,t)$ has $\varh_i$-hop betweenness at most
  $n/2^{L-i+1} \le 2^{i-1}n / h \leq n \varh / h$.  By the union bound,
  this holds for all $s,t \in V$ and all values of $\varh$ with high
  probability.

  As for the running time, for each $i$ and each $r \in R_i$, it takes
  two $\varh_i$-hop distance computations to compute the lengths of
  the edges incident to $r$. The total running time over all $i$ and
  $r \in R_i$ is
  \begin{math}
    \bigO{\sum_{i=0}^L \sizeof{R_i} \eta_i \mu} =
    \bigO{\mu h \log{n}^2}.
  \end{math}
  To compute Johnson's potentials in $H$, observe that every negative
  edge is incident to one of the $\bigO{h \log n}$ sampled vertices in
  $\bigcup_i R_i$, so $\bigO{h \log n}$ hops suffice.  $H$ also has
  $\bigO{m + h n \log n}$ edges. Therefore it takes
  $\bigO{m h \log n + nh^2 \log{n}^2}$ time to compute the Johnson's
  potentials.
\end{proof}

Let $h_0 = \bigO{\log n}$ and let $h \geq h_0$ be a parameter to be
determined.  Let $(s,U,t)$ be a weak $h_0$-hop negative sandwich where
$(s,t)$ has $\parof{\varh+h_0}$-hop betweenness $n\varh/h$ for all
$\varh$. By \reflemma{remotization}, we can reweight the graph with a
valid potential $\psi$ so that $U$ can $\varh$-hop negatively reach at
most $n \varh / h$ vertices for all $\varh$. The goal of the remainder
of this section is to compute an $h$-hop reducer for $G_{U,\psi}$. For
ease of notation we drop all negative vertices outside $U$ and let
$G = G_{U,\psi}$.

We take a bootstrapping approach that starts with a trivial $1$-hop
reducer for the $1$-hop negative reach of $U$, and builds out
$\eta/2$-hop reducers for the $\eta$-hop negative reach for increasing
values of $\eta$. To this end, let $L = \logup[2]{h} + 1$.  For each
$i \in \setof{1,\dots,L-1}$, let $V_i$ be the $2^i$-hop negative reach
of $U$, and let $G_i$ be the subgraph induced by $V_i$. $G_i$ contains
$\bigO{2^in / h}$ vertices and $\bigO{2^i m / h}$ edges. Let $V_L = V$
and $G_L = G$. For each $i$, let $d_i(\cdot, \cdot)$ denote distances
in $G_i$. Let $\outcut{V_i}$ be the directed cut of arcs from $V_i$ to
$V \setminus V_i$.

We construct $2^{i-1}$-hop reducers $H_i$ for each $G_i$
incrementally. Each bootstraps the next through the following distance
estimates:

\begin{definition}
  For $i \in \setof{1,\dots,L}$ and $s,t \in V_i$, we say that a value
  $\delta \in \reals$ is a \defterm{valid distance estimate} for
  $(s,t)$ at level $i$ if it satisfies the following:
  \begin{CompactMathProperties}
  \item \label{valid-1} $\delta \geq d_i(s,t)$.
  \item \label{valid-2} $\delta \leq \phopd{\varh}{s,t}_i$ for all
    $\varh \in [2^{i-1},2^i]$.
  \item \label{valid-3} $\delta + \hopd{0}{t,x}_i + \len{x,y} \geq 0$
    for all arcs $(x,y) \in \outcut{V_i}$
  \end{CompactMathProperties}
\end{definition}

We employ valid distance estimates $\delta_i(s,t)$ at level $i$ for
all $s \in U$ and $t \in \Heads$. \Cref{valid-1,valid-2}
ensure these values represent walks in $G_i$ that are competitive with
any proper $(s,t)$-walk having roughly $2^i$ hops. Counterintuitively,
our techniques struggle to compete with walks having significantly
fewer than $2^i$ hops --- they handle proper walks with more hops more
readily than those with fewer. (The proof of
\reflemma{hop-reduction->valid-estimates} clarifies why this
limitation arises.) \Cref{valid-3} preserves the negative
$2^i$-hop reach of $U$.

The construction is a bootstrapping process built on mutual recursion.
From valid distance estimates at levels below $i$, we construct a
$2^{i-1}$-hop reducer $H_i$ for $G_i$. From this hop reducer $H_i$, we
extract valid distance estimates at level $i$. Each hop reducer
bootstraps the next with double the hop reduction and double the
negative reach, culminating in $H_L$: an $h$-hop reducer for all of
$G$.

Formally we alternate between two lemmas. The first lemma constructs
hop reducers from valid distance estimates.

\begin{lemma}
  \labellemma{valid-estimates->hop-reduction} Let $i \in [L]$. Suppose
  we have a valid distance estimate $\delta_j(s,t)$ at level $j$ for
  all $j < i$, $s \in U$, and $t \in \Heads$. Then one can construct a
  $2^{i-1}$-hop reducer $H_i$ for $G_i$ with $n_i = \bigO{2^i n / h}$
  vertices and $m_i = \bigO{2^i m /h + \sizeof{U}^2 i}$ edges. The
  construction takes $\bigO{m_i + n_i \log n_i}$ time.
\end{lemma}

\begin{proof}
  For $i = 1$, we return $G_1$. Suppose $i > 1$.

  We create an auxiliary graph $H$ starting with disjoint copies of
  $G_j^+$ for all $j \leq i$.  For each $j \leq i$ and each
  $v \in V_j$, let $v_j$ denote the copy of $v$ in $G_j^+$.  For each
  $j < i$, we add arcs between $G_i^+$ and $G_j^+$ as follows.  For
  $s \in U$ and $t \in \Heads$, we add a ``shortcut'' arc $(s_i,t_j)$
  with length $\delta_j(s,t)$. For each edge $(x,y) \in \outcut{V_j}$
  with $x \in V_j, y \in V_i$, we add an ``exit'' arc $(x_j,y_i)$ of
  length $\len{x,y}$.  Lastly, we add ``self-arcs'' $(v_{j},v_i)$ of
  length $0$ for all $j < i$ and $v \in V_j$.

  \begin{center}
    \includegraphics{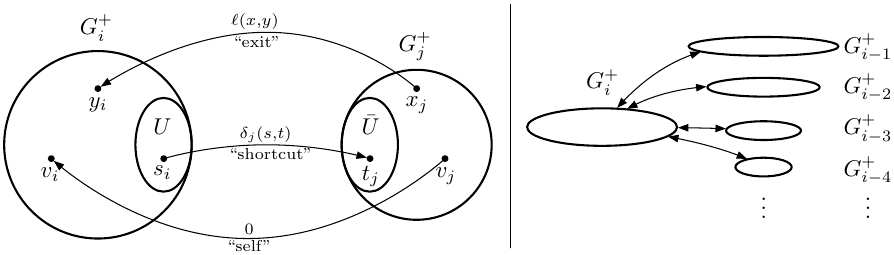}
  \end{center}

  Let $H'$ be the subgraph of $H$ obtained by dropping the self-arcs
  $(v_j,v_i)$ for all $j < i$ and $v \in V_j$.  We claim that for all
  $u, v \in V_i$, $\hopd{1}{u_i,v_i}_{H'} \geq 0$. (I.e., the negative
  edges in $H'$ are \emph{independent} in the sense of
  \cite{Fineman24}.)  Indeed, the only negative edges in $H'$ are of
  the form $(s_i,t_j)$ for $s \in U$, $t \in \Heads$, and $j < i$. Any
  $1$-hop walk $W$ that starts with such an edge $(s_i,t_j)$ and
  returns to $V_i \setminus V_j$ is nonnegative because, letting
  $(x_j,y_i) \in \outcut{V_j}$ be the exit arc where $W$ first leaves
  $V_j$, the prefix of $W$ from $s_i$ to $y_i$ has length at least
  \begin{align*}
    \delta_j(s,t) + \hopd{0}{t,x}_{j} + \len{x,y} \geq 0
  \end{align*}
  by \cref{valid-3} of valid distance estimates.

  Consequently, we can compute the Johnson's potentials
  $\varphi$ for $H'$ in
  $\bigO{m_i + n_i \log n_i}$. Moreover, $\varphi(x_i) = 0$ for all
  $x \in V_i$. Applying $\varphi$ to $H$, the only negative edges in
  $H_{\varphi}$ are the ones omitted from $H'$: namely, the self-arcs
  $(x_j,x_i)$ for all $j < i$ and $x \in V_j$.

  For fixed $j < i$, one can interpret $G_j^+$ and its boundary arcs
  as a compressed version of the layered graph constructed in
  \reflemma{hop-reduction}. The shortcut arcs effectively bypass the
  intermediate layers. The exit arcs and self-arcs here play a similar
  role to those in \reflemma{hop-reduction}. A key difference is that
  the ``$G_j^+$-gadget'' here only addresses proper walks through
  $V_j$ with at least $2^{j-1}$ hops and at most $2^j$ hops, whereas
  the layered graph would naturally handle any $2^j$-hop walk. For
  this reason, we have a $G_j^+$-gadget for every $j < i$, and we will
  have to be more careful when embedding walks from $G_i$ into
  $H_{\varphi}$ in the argument below.

  Formally, we claim that $H_{\varphi}$ is a $2^{i-1}$-hop reducer for
  $G_i$. That means,
  $d_i(s, t) \leq d_{H_\varphi}^{\lceil \eta / 2^{i-1} \rceil}(s_i, t_i) \leq d_i^{\eta}(s, t)$
  for any $\eta$ and $s, t \in V_i$.  For the first inequality, we have
  $d_{H_\varphi}^{\lceil \eta / 2^{i-1} \rceil}(s_i, t_i) \ge d_{H_\varphi}(s_i,t_i) = d_H(s_i,
  t_i) \geq d_{i}(s,t)$ for $s,t \in V_i$ because of \cref{valid-1} of
  valid distance estimates and $\varphi(s_i) = \varphi(t_i) = 0$.

  For the second inequality, consider any proper walk
  $W: s \leadsto t$ in $G_i$ from $s \in V_i$ to $t \in V_i$ and let
  $\varh$ be the number of hops in $W$. We claim there is a
  $\roundup{\varh/2^{i-1}}$-hop walk $W' : s_i \leadsto t_i$ in
  $H_\varphi$ with length at most the length of $W$. It suffices to
  prove the claim for $s \in U$.  We prove the claim by induction on
  the number of hops in $W$, where $\varh=0$ hops is immediate.

  For each index $j$, let $W_j$ be the maximal $2^j$-hop prefix of
  $W$, and let $\bar{W}_j$ be the remaining suffix.  We have two
  cases.

  \subparagraph{Case 1.} Suppose $W_j$ is contained in $G_j$ for all
  $j < i$.  Let $j = \min{\logup[2]{\varh}, i-1}$. We will embed $W_j$
  as a $1$-hop walk via the $G_j^+$-gadget, so to speak, using a
  reset edge to return to the base graph $G_i^+$.

  \begin{center}
    \includegraphics{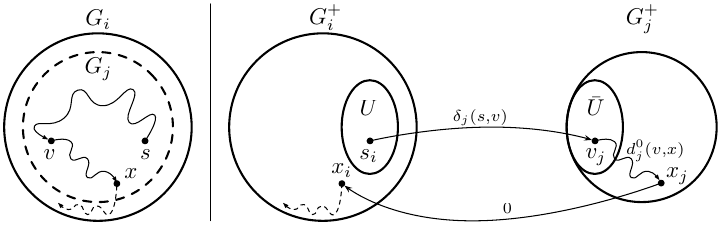}
  \end{center}

  Let $x \in V_j$ be the last vertex in $W_j$. Let
  $\varvarh = \min{\varh,2^{i-1}}$ be the number of hops in $W_j$. Let
  $v \in \Heads$ be the head of the last hop in $W_{j}$.  Consider the
  walk $W_{j}'$ going from $s_i$ to $v_j$ along the shortcut arc
  $(s_i,v_j)$, to $x_{j}$ along the shortest $(v,x)$-path in
  $G_{j}^+$, and to $x_i$ along the self-arc $(x_j,x_i)$. $W_{j}'$ has
  at most one hop, from the last arc $(x_j,x_i)$. Furthermore,
  \begin{align*}
    \len{W_j'}_{H,\varphi}
    &\tago{=} \len{W_j'}_H
      =
      \delta_{j}(s,v) + \hopd{0}{v,x}_j
      \tago{\leq}
      \phopd{\varvarh}{s,v}_{j} + \hopd{0}{v,x}_{j}
      \leq
      \len{W_j}.
  \end{align*}
  Here \tagr is because $\varphi(s_i) = \varphi(x_i) = 0$. \tagr is by
  \cref{valid-2} and $2^{j-1} \leq \varvarh \leq 2^j$.

  Meanwhile, either $W_j = W$; or $j = i-1$, $x \in U$, and
  $\bar{W}_j: x \leadsto t$ is a $(\varh-2^{i-1})$-hop walk in
  $G_i$. In the former, $W_j'$ is a 1-hop $(s_i,t_i)$-walk in
  $H_\varphi$ with length at most the length of $W_j = W$, as
  desired. In the latter, by induction on the number of hops, there is
  a $(\roundup{\varh/2^{i-1}} - 1)$-hop walk $\bar{W}_j'$ from $x_i$
  to $t_i$ with length
  $\len{\bar{W}_j'}_{H,\varphi} \leq \len{\bar{W}_j}$. Concatenated
  together, $W_j'$ and $\bar{W}_{j}'$ gives the desired
  $\roundup{\varh/2^{i-1}}$-hop $(s_i,t_i)$-walk.

  \subparagraph{Case 2.} Suppose $W_j$ is not contained in $G_j$ for
  some $j < i$. Let $j$ be the first index such that $W_{j}$ is not
  contained in $G_j$. Loosely speaking, we will route the prefix of
  $W_j$ contained in $G_j$ as a $0$-hop walk through the
  $G_j^+$-gadget, returning to the base graph $G_i^+$ along an exit
  arc just as $W_j$ leaves $G_j$.

  \begin{center}
    \includegraphics{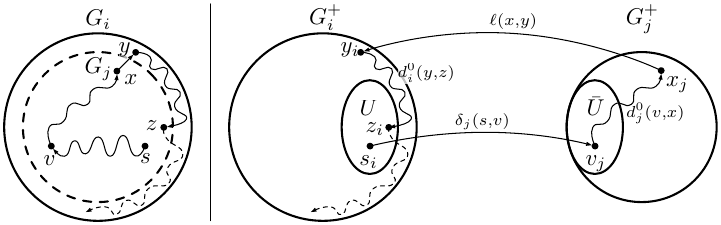}
  \end{center}

  Let $(x,y) \in \outcut{V_j}$ be the arc where $W_j$ first steps out
  of $V_j$. Let $z$ be the first negative vertex after $(x,y)$ if one
  exists; otherwise let $z = t$. Let $W_z: s \leadsto z$ be the prefix
  of $W$ up to $z$ and let $\bar{W}_z: z \leadsto t$ be the remaining
  suffix. Let $\varvarh \in [2^{j-1},2^j]$ be the number of hops in
  $W_z$. Let $u \in \Heads$ be the head of the last hop in $W_z$.
  Consider the walk $W_z': s_i \leadsto z_i$ in $H_\varphi$ going from
  $s_i$ to $u_j$ along the shortcut arc $(s_i,u_j)$, to $x_j$ along
  the shortest $(u, x)$-path in $G_j^+$, to $y_i$ along the exit arc
  $(x_j,y_i)$, and then to $z_i$ along the shortest $(y, z)$-path in
  $G_i^+$.  $W_z'$ has no hops, and
  \begin{align*}
    \len{W_z'}_{H,\varphi}
    &\tago{=} \len{W_z'}_{H} =
      \delta_j(s,u) + \hopd{0}{u,x}_j + \len{x,y} + \hopd{0}{y,z}_i
    \\
    &\tago{\leq}
      \phopd{\varvarh}{s,u} + \hopd{0}{u,x}_j + \len{x,y} +
      \hopd{0}{y,z}_i
      \leq
      \len{W_z},
  \end{align*}
  where \tagr is because $\pote{s_i} = \pote{z_i} = 0$, and \tagr is
  by \cref{valid-2} of valid distance measures.  Meanwhile, either
  $\bar{W}_z$ is empty, or $z \in U$ and $W_z$ is a
  $(\varh - \varvarh)$-hop walk. Either trivially in the former case
  or by induction in the latter, there is a
  $\roundup{(\varh-\varvarh)/2^{i-1}}$-hop walk
  $\bar{W}_z': z_i \leadsto t_i$ in $H_\varphi$ with
  $\len{\bar{W}_z'}_{H,\varphi} \leq \len{\bar{W}_z}$. Together,
  $W_z'$ and $\bar{W}_z'$ give the desired
  $\roundup{\varh/2^{i-1}}$-hop walk in $H_\varphi$.
\end{proof}

The second lemma uses the hop reducer $H_i$ for $G_i$ to create valid
distance estimates at level $i$.

\begin{lemma}
  \labellemma{hop-reduction->valid-estimates} Given a $2^{i-1}$-hop
  reducer $H_i$ for $G_i$ with $m_i$ edges and $n_i $ vertices, one
  can compute valid distance estimates $\delta_i(s,t)$ at level $i$ for all
  $s\in U$ and $t \in \Heads$ with high probability in
  \begin{math}
    \bigO{ \min{1, \log{n} / 2^i} \parof{ |U| \parof{m_i + n_i \log n_i} } }
  \end{math}
  randomized time.
\end{lemma}

\begin{proof}
  For $s \in U$ and $t \in \Heads$, let
  \begin{align*}
    \delta_i(s,t) = \max{\min_{u \in U_0}
    \hopd{2}{s,u}_{H_i} + \hopd{2}{u,t}_{H_i}, -\lambda_i(t)},
  \end{align*}
  where $U_0 \subseteq U$ samples each negative vertex independently
  with probability $\min{1,\bigO{\log n / 2^i}}$, and
  \begin{align*}
    \lambda_i(t) = \inf{\hopd{0}{t,x}_i + \len{x,y} \where (x,y) \in \outcut{V_i}}.
  \end{align*}
  The sampling probability ensures that for any proper $(s,t)$-walk
  with $\bigTheta{2^i}$ hops, we sample at least one vertex from the
  walk with high probability, allowing us to compute distances via
  these midpoints rather than from every $s \in U$. We claim that
  $\delta_i(s,t)$ is a valid distance estimate for $(s,t)$ at level
  $i$ for all $s \in U$ and $t \in \Heads$ with high probability.

  Fix $s \in U$ and $t \in \Heads$.  \Cref{valid-1} follows from the
  fact that $d_{H_i}(x_i,y_i) \geq d_i(x,y)$ for all $x,y \in V_i$.

  Consider \cref{valid-2}.  Let $W$ be the minimum length proper
  $(s,t)$-walk in $G_i$ among all proper walks with at least $2^{i-1}$
  hops and at most $2^i$ hops. (If there is no such walk, then
  \cref{valid-2} automatically holds.)  Since $W$ contains at least
  $2^{i-1}$ distinct negative vertices, with high probability, $U_0$
  contains at least one negative vertex $u$ from $W$. We have
  \begin{align*}
    \hopd{2}{s,u}_{H_i} + \hopd{2}{u,t}_{H_i}
    \tago{\leq}
    \hopd{2^i}{s,u}_{i} + \hopd{2^i}{u,t}_{i}
    \leq
    \len{W}
  \end{align*}
  where \tagr is because $H_i$ is a $2^{i-1}$-hop reducer for $G_i$.
  Meanwhile, for all $(x,y) \in \outcut{V_i}$, we have
  \begin{align*}
    \hopd{2^i}{s,t}_i + \hopd{0}{t,x}_i + \len{x,y} \geq \hopd{2^i}{s,y}
    \geq 0,
  \end{align*}
  because $y$ is not $2^i$-hop negatively reachable from
  $U$. Thus
  \begin{align*}
    \hopd{2^i}{s,t}_i \geq \sup[(x,y) \in \outcut{V_i}]
    -\hopd{0}{t,x}_i - \len{x,y}
    =
    -\lambda_i(t).
  \end{align*}
  Altogether, we have
  \begin{align*}
    \delta_i(s,t) \leq \max{\len{W}, \hopd{2^i}{s,t}_i} \leq
    \phopd{\varh}{s,t}_i
  \end{align*}
  for all $\varh \in [2^{i-1}, 2^i]$.

  Lastly, for \cref{valid-3}, for $(x,y) \in \outcut{X}$, we have
  \begin{align*}
    \delta_i(s,t) + \hopd{0}{t,x}_i + \len{x,y}
    \geq
    \hopd{0}{t,x}_i + \len{x,y} - \lambda_i(t) \geq 0
  \end{align*}
  by choice of $\lambda_i(t)$.

  Thus $\delta_i(s,t)$ is a valid $(s,t)$-distance estimate at level
  $i$ with high probability for fixed $s \in U$ and $t \in
  \Heads$. For each $u \in U_0$, we compute $\hopd{2}{s,u}_{H_i}$ for
  all $s \in U$ and $\hopd{2}{u,t}_{H_i}$ for each $t \in \Heads$ in
  $\bigO{m_i + n_i \log n_i}$ time.  We also compute $\lambda_i(t)$
  for all $t \in \Heads$ by a single Dijkstra computation, by
  reversing $G_i^+$ and introducing an auxiliary vertex $\supersource$
  with arcs $(\supersource, x)$ of length $\len{x,y}$ for every edge
  $(x,y) \in \outcut{V_i}$. For all $t \in U$, $\lambda_i(t)$ is the
  distance from $\supersource$ to $t$ in this auxiliary graph.
\end{proof}

Combining these subroutines, we obtain the following construction for
an $h$-hop reducer for $G_U$.

\begin{lemma}
  \labellemma{bootstrapped-hop-reducer} Let
  $h \geq \bigOmega{\log n}$.  Let $U$ be a set of negative vertices that
  can $\eta$-hop negatively reach at most $n \eta / h$ vertices for all
  $\eta$, with $\sizeof{U} \geq \bigOmega{\log^2 n + h / \log^2 n}$. Then
  one can compute an $h$-hop reducer for $G_U$ with high probability
  in
  \begin{math}
    \bigO{\sizeof{U} \mu \log{n}^2/ h + \sizeof{U}^3 \log^2 \log n}
  \end{math}
  randomized time.
\end{lemma}

\begin{proof}
  We apply
  \cref{lemma:valid-estimates->hop-reduction,lemma:hop-reduction->valid-estimates},
  alternately from $i = 1$ to $L$, omitting the last iteration of
  \reflemma{hop-reduction->valid-estimates}. The $i$th iteration of
  \reflemma{valid-estimates->hop-reduction} uses proper distance
  estimates for levels $j < i$ to create a $2^{i-1}$-hop reducer $H_i$
  for $G_i$, and the $i$th iteration of
  \reflemma{hop-reduction->valid-estimates} uses $H_i$ to create
  proper distance estimates for level $i$ (with high probability).
  The $L$th and final iteration of
  \reflemma{valid-estimates->hop-reduction} produces the desired
  $h$-hop reducer $H = H_L$ for $G_U$.

  For each $i$, $H_i$ has $n_i = \bigO{2^i n / h}$ vertices and
  \begin{math}
    m_i = \bigO{2^i m / h + \sizeof{U}^2 i}
  \end{math}
  edges. Summing over all $i$, the total time spent on
  \reflemma{valid-estimates->hop-reduction} is
  \begin{math}
    \bigO{\sum_{i=1}^{L} m_i + n_i \log n_i} %
    = %
    \bigO{\mu + \sizeof{U}^2 \log^2 n}.
  \end{math}

  For $i_0 = \roundup{\log[2] \log[2] n}$, iterations $1$ through $i_0$
  of \reflemma{hop-reduction->valid-estimates} takes
  \begin{align*}
    \sum_{i=1}^{i_0} \bigO{ |U| \parof{m_i + n_i \log n_i}}
    =
    \sum_{i=1}^{i_0} \bigO{\frac{\sizeof{U} \mu 2^i}{h} + \sizeof{U}^3i}
    =
    \bigO{\frac{\sizeof{U} \mu \log n}{h} + \sizeof{U}^3
    \log^2 \log n}.
  \end{align*}
  Iterations $i_0 + 1$ through $\logup[2]{h}$ takes
  \begin{align*}
    \bigO{\sum_{i=i_0+1}^{\logup[2]{h}} |U| \parof{m_i + n_i \log n_i}}
    &=
      \bigO{\sum_{i=i_0+1}^{\logup[2]{h}} \frac{\sizeof{U} \mu
      \log n}{h} + \frac{\sizeof{U}^3 i \log n}{2^i}} \\
    &=
      \bigO{\frac{\sizeof{U} \mu \log^2 n}{h} +
      \sizeof{U}^3 \sum_{i=1}^{\logup[2]{h} - i_0}
      \frac{i + \log \log n}{2^i}
      }
    \\
    &=
      \bigO{\frac{\sizeof{U} \mu \log^2 n}{h} + \sizeof{U}^3
      \log \log n}
  \end{align*}
  randomized time. Altogether, we spend
  \begin{math}
    \bigO{\sizeof{U} \mu \log{n}^2/ h + \sizeof{U}^3 \log^2
      \log n}
  \end{math}
  time on \reflemma{hop-reduction->valid-estimates}.
\end{proof}

\section[Algorithm for denser graphs]{Nearly $mn^{3/4}$ time in denser
  graphs}

\labelsection{dense}

We now integrate the bootstrapping techniques from
\refsection{bootstrap} into the previous algorithm from \cite{HJQ25}
to obtain faster overall running times for the shortest path
problem. Our algorithm differs depending on the density of the graph
and this section assumes the denser regime where
$m \geq \bigOmega{n^{5/4} \log^2 \log n / \log^{3/4} n}$. We describe
a single iteration of an algorithm that, with high probability,
neutralizes negative vertices in
$\bigO{m \log{n}^{7/4} / \numN^{1/4}}$ amortized time per negative
vertex. Repeating this iteration yields the claimed
$\bigO{m n^{3/4} \log^{7/4} n}$ randomized running time.

Let $h = \numN^{1/4} \log^{1/4} n$ and $h_0 = \bigO{\log n}$. Note
that $m \geq nh \log^2 \log{n} / \log{n}$.  By
\reflemma{multiscale-betweenness-reduction}, we reduce the
$(\varh + h_0)$-hop betweenness to $n \varh / h$, for all
$\varh \leq h$, in
\begin{math}
  \bigO{m h \log^2 n}
\end{math}
time.  By \reflemma{sampling}, in $\bigO{\mu \log n}$ randomized time
and with high probability, compute either (a) Johnson's potentials
neutralizing a subset of $\bigOmega{\sqrt{\numN \log n}}$ negative
edges (as desired), (b) a negative cycle (and return), or (c) a weak
$h_0$-hop negative sandwich $(s,U,t)$ with
$\sizeof{U} = \bigOmega{\sqrt{\numN \log n}}$.  Continuing on in the
latter case, and dropping negative vertices from $U$ as needed, we
assume $\sizeof{U} = \bigTheta{\sqrt{\numN \log n}}$.

By \reflemma{remotization}, we reweight the graph so that $U$ can
$\varh$-hop negatively reach at most $n \varh/h$ vertices for all
$\varh \leq h$ in $\bigO{h_0 \mu}$ time. We then apply
\reflemma{bootstrapped-hop-reducer} to $U$. \reflemma{bootstrapped-hop-reducer} takes
\begin{math}
  \bigO{\sizeof{U} \mu \log{n}^2/h}
\end{math}
randomized time because
\begin{math}
  \sizeof{U}^2 \log^2 \log n \leq \bigO{\numN \log{n} \log^2 \log n} \leq \bigO{m
    \log{n}^2 / h}.
\end{math}
Finally, we compute Johnson's potentials for $G_U$ via $H$ in
\begin{math}
  \bigO{\frac{\sizeof{U}}{h} \parof{\mu + \sizeof{U}^2 \log h}}
  =
  \bigO{\sizeof{U} \mu / h}
\end{math}
time.

Altogether, we spend
\begin{align*}
  \bigO{\mu h \log^2 n + \frac{\sizeof{U}\mu \log^2 n}{h}} = \bigO{\mu \log^2 n \parof{h +
  \frac{\sqrt{\numN \log n}}{h}}}
\end{align*}
time to neutralize $U$. The parameter $h$ balances these two terms,
giving $\bigO{m \log{n}^2/ h} = \bigO{m \log{n}^{7/4} / \numN^{1/4}}$
amortized time per negative vertex neutralized.

\section[Algorithm for sparser graphs]{Nearly $m^{4/5} n$ time in
  sparser graphs}

\labelsection{sparse}

We now turn to the page to sparse graphs with $m \leq
\apxO{n^{5/4}}$. The previous algorithm from \refsection{dense} has
two bottlenecks in this regime that we will address.

First, the betweenness reduction introduced in
\reflemma{multiscale-betweenness-reduction} for bootstrapping hop
reducers spends $\apxO{mh + nh^2}$ time. When the graph becomes
sparse, it takes $\apxO{nh^2} = \apxO{n \sqrt{k}}$ time to eventually
neutralize $\apxO{\sqrt{k}}$ edges --- yielding $\apxO{n}$ time per
edge instead of $\apxO{m / h}$.

Second, the hop reduction graphs
$H_i$ have $\sizeof{U}^2 = \bigOmega{\numN}$ shortcut edges, which for
$i$ small and $k \approx n$ dominates the $\bigO{m 2^i / h}$ edges
from the underlying graph --- incurring $\apxO{n}$ time per negative
edge when computing distances through $H_i$.

\paragraph{Sparse betweenness reduction.}
Consider the first bottleneck from betweenness reduction, which
currently takes $\apxO{mh + nh^2}$ time. We want to eliminate the
$\apxO{nh^2}$ term.

The $\apxO{nh^2}$ term comes from an $h$-hop distance computation in a
$\apxO{m + nh}$-edge auxiliary graph. The $\apxO{nh}$ edges breaks
down as $\apxO{h}$ edges per negative vertex, and every negative
vertex has these auxiliary edges in order to ensure that all pairs of
negative vertices have low betweenness.

The analysis only requires low betweenness for the negative vertices
that are sampled.  Therefore, instead of lowering the betweenness for
all pairs of negative vertices before sampling, we first sample the negative
vertices, and then lower the betweenness only for pairs of sampled
vertices. Limiting the betweenness reduction to sampled pairs
removes the $\apxO{nh}$ extra edges from the auxiliary graph, and the
$\apxO{nh^2}$ term from the running time.

However, the betweenness now depends on the sampled vertices, and we
can no longer apply \reflemma{sampling} as a black box. The following
lemma integrates
\cref{lemma:betweenness-reduction,lemma:implicit-proper-distance,lemma:sampling}
together with a more aggressive probabilistic analysis to achieve the
same overall effect in much less time.

\begin{lemma}
  \labellemma{sparse-betweenness-reduction} Let
  $h_0 = \bigTheta{\log n}$, $h \geq h_0$, and
  $q \leq \bigO{\numN / \log{n}}$. One can compute, with high probability
  in $\bigO{h \log{n}^2 \parof{\mu + h \numN / q}}$ randomized
  time, either:
  \begin{compactmathresults}
  \item A negative cycle.
  \item Valid potentials $\varphi$ neutralizing a set of
    $\bigOmega{\numN/q}$ negative vertices.
  \item Valid potentials $\varphi$ and a weak $h_0$-hop negative sandwich
    $(s,U,t)$ in $G_{\varphi}$ such that:
    \begin{compactmathproperties}
    \item $(s,t)$ has $(\varh+h_0)$-hop betweenness $n\varh/h$ for all $\varh$.
    \item $\sizeof{U} \geq \bigOmega{q h_0}$.
    \end{compactmathproperties}
  \end{compactmathresults}
\end{lemma}

\begin{proof}
  Let $S \subseteq N$ sample each negative vertex independently with
  probability $1/q$.  $\sizeof{S} = \bigOmega{\numN/q}$ with high
  probability, and henceforth we assume this is the case. We repeat
  the same steps as in the proof of
  \reflemma{multiscale-betweenness-reduction}, except we only add arcs
  incident to $S$ when building the auxiliary graph $H$. $H$ now has
  $\bigO{m + h k \log{n} / q}$ edges, and computing the Johnson's
  potentials $\varphi$ takes only
  $\bigO{h \log{n} \parof{m + hk \log{n} / q}}$ time.  For all
  $r \in R_i$ and $s \in S$, we have $\hopdp{2^i}{r,s} \geq 0$ and
  $\hopdp{2^i}{s,r} \geq 0$.

  Applying \reflemma{implicit-proper-distance} to $S$, in
  $\bigO{h_0 \mu}$ time, we obtain either (a) a negative cycle, (b)
  potentials neutralizing $S$, or (c) a pair $s,t \in V$ with
  $\phopd{h_0}{s,t} < 0$. In the latter case, we claim that with high
  probability, the weak $h_0$-hop negative sandwich
  $U = \setof{x \in N \where \hopdp{h_0}{s,x} + \hopdp{h_0}{x,t} \leq
    0}$ has $\bigOmega{q h_0}$ vertices.

  For every pair $s,t \in N$, and value $\lambda$, let
  $U_{\lambda}(s,t) = \setof{x \in N \where \hopd{h_0}{s,x} +
    \hopd{h_0}{x,t} \leq \lambda}$. We call $U_{\lambda}(s,t)$ a
  \defterm{sandwich}.  As $\lambda$ ranges from $-\infty$ to
  $+\infty$, $U_{\lambda}(s,t)$ is monotonically increasing, and takes
  on at most $\numN$ distinct sets.  Call the sandwich
  $U_{\lambda}(s,t)$ \defterm{small} if
  $\sizeof{U_{\lambda}(s,t)} \leq c q h_0$ for a sufficiently small
  constant $c$, and \defterm{big} otherwise.  We say that the random
  sample $S$ \emph{hits} a sandwich $U_{\lambda}(s,t)$ if $S$ contains
  at least $h_0 - 1$ vertices from $U_{\lambda}(s,t)$.  For any small
  sandwich $U_{\lambda}(s,t)$, $S$ expects to contain at most $c h_0$
  vertices from $U_{\lambda}(s,t)$. Since $h_0 = \bigOmega{\log n}$,
  by standard concentration bounds, $S$ does not hit
  $U_{\lambda}(s,t)$ with high probability. Taking the union bound
  over all $\bigO{\numN^3}$ small sandwiches, $S$ only hits big
  sandwiches with high probability. Henceforth assume this is the
  case.

  Consider the negative sandwich $U$ between $s$ and $t$ returned by
  \reflemma{implicit-proper-distance}. Observe that
  \begin{align*}
    U = \setof{x \in N \where \pote{s} + \hopd{h_0}{s,x} + \hopd{h_0}{x,t} -
    \pote{t} \leq 0}
    =
    U_{\pote{t} - \pote{s}}(s,t).
  \end{align*}
  Additionally, $\phopdp{h_0}{s,t} < 0$ implies that $S$ hits
  $U$. Thus $U$ is a big sandwich with $\sizeof{U} = \bigOmega{qh_0}$
  vertices.
\end{proof}

\paragraph{Putting it all together.}
The remaining bottleneck comes from the $\apxOmega{\sizeof{U}^2}$
shortcut edges dominating the size of the smaller hop-reducers
$H_i$. We do not modify our construction to address this bottleneck,
but instead toggle the parameters to balance out the $\sizeof{U}^2$
term.

Assume $m \leq \bigO{n^{5/4} \log^{5/4} n}$. Let
$h_0 = \bigO{\log n}$.  We repeat the same iteration described in
\refsection{dense}, except with
\reflemma{sparse-betweenness-reduction} in place of
\cref{lemma:betweenness-reduction,lemma:implicit-proper-distance,lemma:sampling},
and different choices of $h$ depending on $\numN$.

In the first regime we have
$\mu \leq \numN^{5/4} \log^2 \log{n} / \log^{3/4} n$.  Let
$h = \parof{\mu \log{n}^2 / \log^2 \log n}^{1/5}$ and
$q = \sqrt{\mu / h \log^2 \log n}$.
\Reflemma{sparse-betweenness-reduction} takes $\bigO{h \mu \log^2 n}$
time.  Because $k / q \geq q h_0$, the bottleneck is in the third
outcome of \reflemma{sparse-betweenness-reduction}, where we obtain an
$h_0$-hop negative sandwich $(s,U,t)$ of size
$\sizeof{U} \geq \bigOmega{qh_0}$.  We assume
$\sizeof{U} = \bigTheta{q h_0}$ by dropping vertices as needed.  Note
that by choice of $q$,
\begin{math}
  \sizeof{U}^2 %
  = \bigTheta{q^2 h_0^2} %
  = \bigTheta{\mu \log{n}^2 / h \log^2 \log n}. %
\end{math}
By \reflemma{remotization}, we reweight the graph so that $U$ can
$\varh$-hop negatively reach at most $n \varh/h$ vertices for all
$\varh \leq h$ in $\bigO{h_0 \mu}$ time. We then apply
\reflemma{bootstrapped-hop-reducer} to $U$ in
\begin{math}
  \bigO{\sizeof{U} \mu \log^2n / {h}}
\end{math}
time. Finally, we use $H$ to compute Johnson's potentials for
$G_U$ in
\begin{math}
  \bigO{\frac{\sizeof{U}}{h} \parof{\mu + \sizeof{U}^2 \log h }}
\end{math}
time.  Overall, we spend
$\bigO{\mu \log^2 n \parof{h + \sizeof{U} / h}}$ time neutralizing
$U$. As $h$ balances these terms, we conclude that we spend
$\bigO{\mu \log^2 n / h} = \bigO{\mu^{4/5} \log^{8/5} n \log^{2/5}
  \log n}$ time per neutralized vertex.

In the second regime, we have
$\mu \geq \numN^{5/4} \log^2 \log n / \log^{3/4} n$. Let
$h = \numN^{1/4} \log^{1/4} n$ and $q = \sqrt{\numN / \log n}$, the
same as in the denser setting of \refsection{dense}. Retracing the
same calculations from \refsection{dense} shows that we either return
a negative cycle or spend
$\bigO{h \mu \log{n}^2 + \sizeof{U} \mu \log^2 n / h}$ time to
neutralize $\bigOmega{\sqrt{\numN \log n}}$ negative edges. As before,
$h$ balances these two terms, giving
$\bigO{\mu \log{n}^{7/4} / \numN^{1/4}}$ time per neutralized edge.

Any value of $\numN$ falls in either of the two regimes. Thus, by one
method or another, we are able to neutralize vertices at a rate of
\begin{math}
  \bigO{\mu \log{n}^{7/4} / \numN^{1/4} + \mu^{4/5} \log^{8/5} n
    \log^{2/5} \log n}
\end{math}
randomized time per negative vertex. Repeating this until (almost)
all negative edges are neutralized takes
\begin{math}
  \bigO{\mu n^{3/4} \log{n}^{7/4} + \mu^{4/5} n \log^{8/5} n \log^{2/5}
    \log n}
\end{math}
randomized time.

\begin{remark}
  We did not try to optimize logarithmic factors in favor of a simpler
  exposition. We believe they can be improved. Nonetheless we have
  accounted for and balanced the logarithmic factors of the presented
  algorithms for the interested reader.
\end{remark}

\printbibliography


\end{document}
